\documentclass[11pt,letterpaper]{article}

\usepackage[utf8]{inputenc}
\usepackage[T1]{fontenc}
\usepackage{lmodern}
\usepackage[DIV=11]{typearea} 

\usepackage{microtype}

\usepackage{amsmath}
\usepackage{amssymb}
\usepackage{amsthm}
\usepackage{thmtools}

\usepackage[procnumbered,ruled,vlined,linesnumbered]{algorithm2e}

\usepackage{xcolor}
\usepackage{xspace}

\usepackage{tikz}
\usepackage{gnuplot-lua-tikz}

\usepackage[
	style=alphabetic,
	backref=true,
	doi=false,
	url=false,
	maxcitenames=3,
	mincitenames=3,
	maxbibnames=10,
	minbibnames=10,
	backend=bibtex8,
	sortlocale=en_US
]{biblatex}

\usepackage[ocgcolorlinks]{hyperref} 
\usepackage{cleveref}

\newcommand*\samethanks[1][\value{footnote}]{\footnotemark[#1]}

\allowdisplaybreaks[1]

\makeatletter
\g@addto@macro\bfseries{\boldmath}
\makeatother

\makeatletter
\g@addto@macro\mdseries{\unboldmath}
\g@addto@macro\normalfont{\unboldmath}
\g@addto@macro\rmfamily{\unboldmath}
\g@addto@macro\upshape{\unboldmath}
\makeatother

\DeclareCiteCommand{\citem}
    {}
    {\mkbibbrackets{\bibhyperref{\usebibmacro{postnote}}}}
    {\multicitedelim}
    {}

\renewcommand*{\multicitedelim}{\addcomma\space}

\AtEveryCitekey{\clearfield{note}}

\newcommand{\myhref}[1]{%
  \iffieldundef{doi}
    {\iffieldundef{url}
       {#1}
       {\href{\strfield{url}}{#1}}}
    {\href{http://dx.doi.org/\strfield{doi}}{#1}}%
}

\DeclareFieldFormat{title}{\myhref{\mkbibemph{#1}}}
\DeclareFieldFormat
  [article,inbook,incollection,inproceedings,patent,thesis,unpublished]
  {title}{\myhref{\mkbibquote{#1\isdot}}}

\makeatletter
\AtBeginDocument{%
    \newlength{\temp@x}%
    \newlength{\temp@y}%
    \newlength{\temp@w}%
    \newlength{\temp@h}%
    \def\my@coords#1#2#3#4{%
      \setlength{\temp@x}{#1}%
      \setlength{\temp@y}{#2}%
      \setlength{\temp@w}{#3}%
      \setlength{\temp@h}{#4}%
      \adjustlengths{}%
      \my@pdfliteral{\strip@pt\temp@x\space\strip@pt\temp@y\space\strip@pt\temp@w\space\strip@pt\temp@h\space re}}%
    \ifpdf
      \typeout{In PDF mode}%
      \def\my@pdfliteral#1{\pdfliteral page{#1}}
      \def\adjustlengths{}%
    \fi
    \ifxetex
      \def\my@pdfliteral #1{}
      \def\adjustlengths{\setlength{\temp@h}{-\temp@h}\addtolength{\temp@y}{1in}\addtolength{\temp@x}{-1in}}%
    \fi%
    \def\Hy@colorlink#1{%
      \begingroup
        \ifHy@ocgcolorlinks
          \def\Hy@ocgcolor{#1}%
          \my@pdfliteral{q}%
          \my@pdfliteral{7 Tr}
        \else
          \HyColor@UseColor#1%
        \fi
    }%
    \def\Hy@endcolorlink{%
      \ifHy@ocgcolorlinks%
        \my@pdfliteral{/OC/OCPrint BDC}%
        \my@coords{0pt}{0pt}{\pdfpagewidth}{\pdfpageheight}%
        \my@pdfliteral{F}
        \my@pdfliteral{EMC/OC/OCView BDC}%
        \begingroup%
          \expandafter\HyColor@UseColor\Hy@ocgcolor%
          \my@coords{0pt}{0pt}{\pdfpagewidth}{\pdfpageheight}%
          \my@pdfliteral{F}
        \endgroup%
        \my@pdfliteral{EMC}%
        \my@pdfliteral{0 Tr}
        \my@pdfliteral{Q}%
      \fi
      \endgroup
    }%
}
\makeatother

\colorlet{DarkRed}{red!50!black}
\colorlet{DarkGreen}{green!50!black}
\colorlet{DarkBlue}{blue!50!black}

\hypersetup{
	linkcolor = DarkRed,
	citecolor = DarkGreen,
	urlcolor = DarkBlue,
	bookmarks = true,
	bookmarksnumbered = true,
	linktocpage = true
}

\declaretheorem[numberwithin=section]{theorem}
\declaretheorem[numberlike=theorem]{lemma}
\declaretheorem[numberlike=theorem]{proposition}

\declaretheorem[numberlike=theorem]{definition}

\crefname{algorithm}{Procedure}{Procedures}
\Crefname{algorithm}{Procedure}{Procedures}

\DontPrintSemicolon
\SetKw{KwAnd}{and}
\SetProcNameSty{textsc}
\SetFuncSty{textsc}
\SetKwFunction{ApproximatePathUnion}{ApproximatePathUnion}

\newcommand{\dist}{\ensuremath{d}}

\newcommand{\cP}{\mathcal{P}}
\newcommand{\cQ}{Q}

\title{Improved Algorithms for Decremental Single-Source Reachability on Directed Graphs\thanks{This paper was presented at the International Colloquium on Automata, Languages and Programming (ICALP) 2015. A full version combining the findings of this paper and its predecessor~\cite{HenzingerKNSTOC14} is available at \url{http://arxiv.org/abs/1504.07959}.}}
\author{Monika Henzinger\thanks{University of Vienna, Faculty of Computer Science, Austria. Supported by the Austrian Science Fund (FWF): P23499-N23 and the University of Vienna (IK \mbox{I049-N}). The research leading to these results has received funding from the European Research Council under the European Union's Seventh Framework Programme (FP/2007-2013) / ERC Grant Agreement no. 340506.}
 \and Sebastian Krinninger\samethanks[2]
 \and Danupon Nanongkai\thanks{KTH Royal Institute of Technology, Sweden. Work partially done while at University of Vienna, Faculty of Computer Science, Austria.}
}

\date{}

\hypersetup{
	pdftitle = {Improved Algorithms for Decremental Single-Source Reachability on Directed Graphs},
	pdfauthor = {Monika Henzinger, Sebastian Krinninger, Danupon Nanongkai}
}

\begin{document}
\maketitle
\begin{abstract}
Recently we presented the first algorithm for maintaining the set of nodes reachable from a source node in a directed graph that is modified by edge deletions with $o(mn)$ total update  time,  where $m$ is the number of edges and $n$ is the number of nodes in the graph \citem[Henzinger et al.\ STOC 2014]{HenzingerKNSTOC14}.
The algorithm is a combination of several different algorithms, each for a different $m$ vs.~$n$ trade-off.
For the case of $m = \Theta(n^{1.5})$ the running time is $O(n^{2.47})$, just barely below $mn = \Theta(n^{2.5})$.
In this paper we simplify the previous algorithm using new algorithmic ideas and achieve an improved running time of $\tilde O(\min( m^{7/6} n^{2/3}, m^{3/4} n^{5/4 + o(1)}, m^{2/3} n^{4/3+o(1)} + m^{3/7} n^{12/7+o(1)}))$. 
This gives, e.g., $O(n^{2.36})$ for the notorious case $m = \Theta(n^{1.5})$.
We obtain the same upper bounds for the problem of maintaining the strongly connected components of a directed graph undergoing edge deletions.
Our algorithms are correct with high probabililty against an oblivious adversary.

\end{abstract}
\newpage


\section{Introduction}

In this paper we study the decremental reachability problem.
Given a directed graph $ G $ with $n$ nodes and $m$ edges and a source node $ s $ in $ G $ a {\em decremental single-source reachability algorithm} 
maintains the set of nodes reachable from $ s $ (i.e., all nodes $ v $ for which there is a path from $ s $ to $ v $ in the current version of $ G $) during a sequence of edge deletions.
The goal is to minimize the \emph{total update time}, i.e., the total time needed to process \emph{all} deletions such that reachability queries can be answered in constant time.
A {\em decremental $s$-$t$ reachability algorithm} is given a graph $ G $ undergoing edge deletions, a source node $ s $, and a sink node $ t $
and it determines after every deletion in $ G $ whether $ s $ can still reach $ t $.

\paragraph*{Related Work.}
The incremental version of the single-source reachability problem, in which edges are \emph{inserted} into the graph, can be solved with a total update time of $ O(m) $ by performing an incremental graph search, where $ m $ is the final number of edges.
Italiano~\cite{Italiano88} showed that in directed acyclic graphs the decremental problem can be solved in time $ O (m) $ as well.
In general directed graphs however, the problem could for a long time only be solved in time $ O (m n) $ using the more general decremental single-source shortest paths algorithm of Even and Shiloach~\cite{EvenS81,HenzingerK95,King99},
which maintains a breadth-first search tree rooted at $s$,
called {\em ES-tree}.
This upper bound of $ O (m n) $ is also achieved for the seemingly more complex decremental \emph{all-pairs} reachability problem (also known as transitive closure)~\cite{RodittyZ08,Lacki13}.
In the fully dynamic version of single-source reachability both insertions and deletions of edges are possible.
The matrix-multiplication based transitive closure algorithms of Sankowski~\cite{Sankowski04} give fully dynamic algorithms for single-source reachability and $s$-$t$ reachability with worst-case running times of $ O (n^{1.575}) $ and $ O (n^{1.495}) $ \emph{per update}, respectively.

These upper bounds have recently been complemented by Abboud and Vassilevska Williams \cite{AbboudW14} as follows.
For the decremental $s$-$t$ reachability problem, a combinatorial algorithm with a \emph{worst-case} running time of $ O (n^{2-\delta}) $ (for some $ \delta > 0 $) per update or query implies a faster combinatorial algorithm for Boolean matrix multiplication and, as has been shown by Vassilevska Williams and Williams~\cite{WilliamsW10}, for other problems as well.
(For non-combinatorial algorithms, Henzinger~et~al.~\cite{HenzingerKNS15} showed that there is no algorithm with worst-case $ O (n^{1-\delta}) $ update and $ O (n^{2-\delta}) $ query time, assuming the so-called Online Matrix-Vector Multiplication conjecture.) 
Furthermore, for the problem of maintaining the number of nodes reachable from a source under deletions (which our algorithms can do) a worst-case running time of $ O (m^{1-\delta}) $ (for some $ \delta > 0 $) per update or query falsifies the strong exponential time hypothesis.
Thus, amortization is indeed necessary to bypass these bounds.

In~\cite{HenzingerKNSTOC14} we recently improved upon the long-standing upper bound of $ O (m n) $ for decremental single-source reachability in directed graphs.
In particular, we developed several algorithms whose combined expected running time is polynomially faster than $ O (m n) $ for all values of $ m $ (i.e., for all possible densities of the initial graph).
By a reduction from single-source reachability, our results in~\cite{HenzingerKNSTOC14} immediately give an $ o (mn) $ algorithm for maintaining strongly connected components under edge deletions.
Previously, the fastest decremental algorithms for this problem had a total update time of $ O(m n) $ as well \cite{RodittyZ08,Lacki13,Roditty13}.

\paragraph*{Our Results.}
In this paper we improve upon the upper bounds provided in~\cite{HenzingerKNSTOC14}.
Furthermore, the running times achieved in this paper are arguably more natural than those in~\cite{HenzingerKNSTOC14}.
Although we previously broke the $ O (m n) $ barrier for all values of $ m $, we barely did so, giving a bound of $O(n^{2.47})$, when $ m = \Theta (n^{1.5}) $.
In this paper we also get a better improvement, namely $O(n^{2.36})$  in this notorious case.
In general, we can combine the algorithms of this paper to obtain a running time of $ O (m n^{0.9 + o(1)}) $, whereas in~\cite{HenzingerKNSTOC14} we obtained $ \tilde O (mn^{0.984}) $.

In~\cite{HenzingerKNSTOC14} the starting point was to solve the decremental $s$-$t$ reachability problem, which is also the case here.
For this problem we obtain two algorithms with total update times of $ \tilde O (\min (m^{5/4} n^{1/2}, m^{2/3} n^{4/3 + o(1)})) $ and $ O (m^{2/3} n^{4/3+o(1)} + m^{3/7} n^{12/7+o(1)}) $, respectively.
Just as in~\cite{HenzingerKNSTOC14}, extensions of these algorithms solve the decremental single-source reachability problem with total update times of $ \tilde O (\min (m^{7/6} n^{2/3}, m^{3/4} n^{5/4 + o(1)})) $ and $ O (m^{2/3} n^{4/3+o(1)} + m^{3/7} n^{12/7+o(1)}) $, respectively.
Furthermore, it follows from a reduction~\cite{RodittyZ08,HenzingerKNSTOC14} that there are algorithms for the decremental strongly connected components problem whose running times are the same up to a logarithmic factor.
We compare these new results to the ones of~\cite{HenzingerKNSTOC14} in Figure~\ref{fig:comparison_running_time}.
All our algorithms are correct with high probability who fixes its sequence of updates and queries
before the algorithm is initialized and their running time bounds hold in expectation.
Due to space constraints this paper only contains an overview of the algorithm that has a total update time of $ O (m^{2/3} n^{4/3+o(1)} + m^{3/7} n^{12/7+o(1)}) $ and is thus the current fastest for dense graphs.
The other algorithm and all omitted proofs can be found in the full version of this paper.

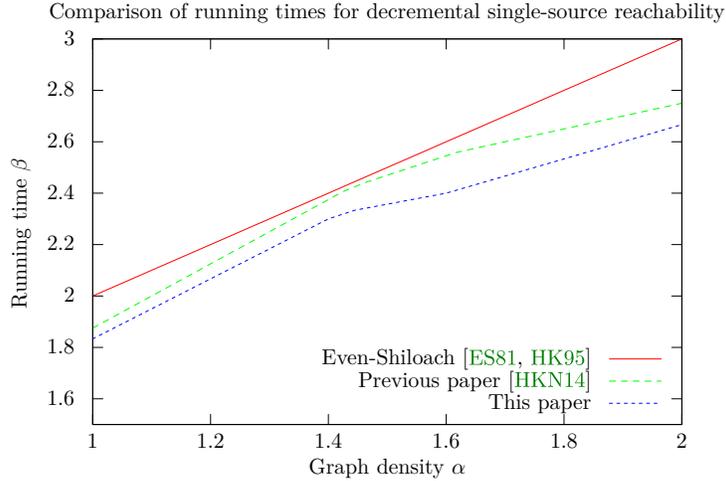
\begin{figure}[htbp]
\centering
\scalebox{0.75}{
\begin{tikzpicture}[gnuplot]
\path (0.000,0.000) rectangle (12.500,8.750);
\gpcolor{color=gp lt color border}
\gpsetlinetype{gp lt border}
\gpsetlinewidth{1.00}
\draw[gp path] (1.504,1.441)--(1.684,1.441);
\draw[gp path] (11.947,1.441)--(11.767,1.441);
\node[gp node right] at (1.320,1.441) { 1.6};
\draw[gp path] (1.504,2.353)--(1.684,2.353);
\draw[gp path] (11.947,2.353)--(11.767,2.353);
\node[gp node right] at (1.320,2.353) { 1.8};
\draw[gp path] (1.504,3.265)--(1.684,3.265);
\draw[gp path] (11.947,3.265)--(11.767,3.265);
\node[gp node right] at (1.320,3.265) { 2};
\draw[gp path] (1.504,4.177)--(1.684,4.177);
\draw[gp path] (11.947,4.177)--(11.767,4.177);
\node[gp node right] at (1.320,4.177) { 2.2};
\draw[gp path] (1.504,5.089)--(1.684,5.089);
\draw[gp path] (11.947,5.089)--(11.767,5.089);
\node[gp node right] at (1.320,5.089) { 2.4};
\draw[gp path] (1.504,6.001)--(1.684,6.001);
\draw[gp path] (11.947,6.001)--(11.767,6.001);
\node[gp node right] at (1.320,6.001) { 2.6};
\draw[gp path] (1.504,6.913)--(1.684,6.913);
\draw[gp path] (11.947,6.913)--(11.767,6.913);
\node[gp node right] at (1.320,6.913) { 2.8};
\draw[gp path] (1.504,7.825)--(1.684,7.825);
\draw[gp path] (11.947,7.825)--(11.767,7.825);
\node[gp node right] at (1.320,7.825) { 3};
\draw[gp path] (1.504,0.985)--(1.504,1.165);
\draw[gp path] (1.504,7.825)--(1.504,7.645);
\node[gp node center] at (1.504,0.677) { 1};
\draw[gp path] (3.593,0.985)--(3.593,1.165);
\draw[gp path] (3.593,7.825)--(3.593,7.645);
\node[gp node center] at (3.593,0.677) { 1.2};
\draw[gp path] (5.681,0.985)--(5.681,1.165);
\draw[gp path] (5.681,7.825)--(5.681,7.645);
\node[gp node center] at (5.681,0.677) { 1.4};
\draw[gp path] (7.770,0.985)--(7.770,1.165);
\draw[gp path] (7.770,7.825)--(7.770,7.645);
\node[gp node center] at (7.770,0.677) { 1.6};
\draw[gp path] (9.858,0.985)--(9.858,1.165);
\draw[gp path] (9.858,7.825)--(9.858,7.645);
\node[gp node center] at (9.858,0.677) { 1.8};
\draw[gp path] (11.947,0.985)--(11.947,1.165);
\draw[gp path] (11.947,7.825)--(11.947,7.645);
\node[gp node center] at (11.947,0.677) { 2};
\draw[gp path] (1.504,7.825)--(1.504,0.985)--(11.947,0.985)--(11.947,7.825)--cycle;
\node[gp node center,rotate=-270] at (0.246,4.405) {Running time $\beta$};
\node[gp node center] at (6.725,0.215) {Graph density $\alpha$};
\node[gp node center] at (6.725,8.287) {Comparison of running times for decremental single-source reachability};
\node[gp node right] at (10.479,2.148) {Even-Shiloach \cite{EvenS81,HenzingerK95}};
\gpcolor{color=gp lt color 0}
\gpsetlinetype{gp lt plot 0}
\draw[gp path] (10.663,2.148)--(11.579,2.148);
\draw[gp path] (1.504,3.265)--(1.609,3.311)--(1.715,3.357)--(1.820,3.403)--(1.926,3.449)%
  --(2.031,3.495)--(2.137,3.541)--(2.242,3.587)--(2.348,3.633)--(2.453,3.680)--(2.559,3.726)%
  --(2.664,3.772)--(2.770,3.818)--(2.875,3.864)--(2.981,3.910)--(3.086,3.956)--(3.192,4.002)%
  --(3.297,4.048)--(3.403,4.094)--(3.508,4.140)--(3.614,4.186)--(3.719,4.232)--(3.825,4.278)%
  --(3.930,4.324)--(4.036,4.370)--(4.141,4.417)--(4.247,4.463)--(4.352,4.509)--(4.458,4.555)%
  --(4.563,4.601)--(4.669,4.647)--(4.774,4.693)--(4.880,4.739)--(4.985,4.785)--(5.090,4.831)%
  --(5.196,4.877)--(5.301,4.923)--(5.407,4.969)--(5.512,5.015)--(5.618,5.061)--(5.723,5.107)%
  --(5.829,5.153)--(5.934,5.200)--(6.040,5.246)--(6.145,5.292)--(6.251,5.338)--(6.356,5.384)%
  --(6.462,5.430)--(6.567,5.476)--(6.673,5.522)--(6.778,5.568)--(6.884,5.614)--(6.989,5.660)%
  --(7.095,5.706)--(7.200,5.752)--(7.306,5.798)--(7.411,5.844)--(7.517,5.890)--(7.622,5.937)%
  --(7.728,5.983)--(7.833,6.029)--(7.939,6.075)--(8.044,6.121)--(8.150,6.167)--(8.255,6.213)%
  --(8.361,6.259)--(8.466,6.305)--(8.571,6.351)--(8.677,6.397)--(8.782,6.443)--(8.888,6.489)%
  --(8.993,6.535)--(9.099,6.581)--(9.204,6.627)--(9.310,6.673)--(9.415,6.720)--(9.521,6.766)%
  --(9.626,6.812)--(9.732,6.858)--(9.837,6.904)--(9.943,6.950)--(10.048,6.996)--(10.154,7.042)%
  --(10.259,7.088)--(10.365,7.134)--(10.470,7.180)--(10.576,7.226)--(10.681,7.272)--(10.787,7.318)%
  --(10.892,7.364)--(10.998,7.410)--(11.103,7.457)--(11.209,7.503)--(11.314,7.549)--(11.420,7.595)%
  --(11.525,7.641)--(11.631,7.687)--(11.736,7.733)--(11.842,7.779)--(11.947,7.825);
\gpcolor{color=gp lt color border}
\node[gp node right] at (10.479,1.755) {Previous paper \cite{HenzingerKNSTOC14}};
\gpcolor{color=gp lt color 1}
\gpsetlinetype{gp lt plot 1}
\draw[gp path] (10.663,1.755)--(11.579,1.755);
\draw[gp path] (1.504,2.695)--(1.609,2.753)--(1.715,2.810)--(1.820,2.868)--(1.926,2.925)%
  --(2.031,2.983)--(2.137,3.040)--(2.242,3.098)--(2.348,3.156)--(2.453,3.213)--(2.559,3.271)%
  --(2.664,3.328)--(2.770,3.386)--(2.875,3.443)--(2.981,3.501)--(3.086,3.559)--(3.192,3.616)%
  --(3.297,3.674)--(3.403,3.731)--(3.508,3.789)--(3.614,3.847)--(3.719,3.904)--(3.825,3.962)%
  --(3.930,4.019)--(4.036,4.077)--(4.141,4.134)--(4.247,4.192)--(4.352,4.250)--(4.458,4.307)%
  --(4.563,4.365)--(4.669,4.422)--(4.774,4.480)--(4.880,4.537)--(4.985,4.595)--(5.090,4.653)%
  --(5.196,4.710)--(5.301,4.768)--(5.407,4.825)--(5.512,4.883)--(5.618,4.940)--(5.723,4.998)%
  --(5.829,5.056)--(5.934,5.113)--(6.040,5.159)--(6.145,5.201)--(6.251,5.244)--(6.356,5.285)%
  --(6.462,5.320)--(6.567,5.354)--(6.673,5.389)--(6.778,5.424)--(6.884,5.459)--(6.989,5.494)%
  --(7.095,5.529)--(7.200,5.564)--(7.306,5.599)--(7.411,5.634)--(7.517,5.669)--(7.622,5.703)%
  --(7.728,5.738)--(7.833,5.773)--(7.939,5.808)--(8.044,5.833)--(8.150,5.856)--(8.255,5.879)%
  --(8.361,5.902)--(8.466,5.925)--(8.571,5.948)--(8.677,5.971)--(8.782,5.994)--(8.888,6.017)%
  --(8.993,6.040)--(9.099,6.063)--(9.204,6.086)--(9.310,6.109)--(9.415,6.132)--(9.521,6.155)%
  --(9.626,6.178)--(9.732,6.201)--(9.837,6.224)--(9.943,6.247)--(10.048,6.270)--(10.154,6.293)%
  --(10.259,6.317)--(10.365,6.340)--(10.470,6.363)--(10.576,6.386)--(10.681,6.409)--(10.787,6.432)%
  --(10.892,6.455)--(10.998,6.478)--(11.103,6.501)--(11.209,6.524)--(11.314,6.547)--(11.420,6.570)%
  --(11.525,6.593)--(11.631,6.616)--(11.736,6.639)--(11.842,6.662)--(11.947,6.685);
\gpcolor{color=gp lt color border}
\node[gp node right] at (10.479,1.362) {This paper};
\gpcolor{color=gp lt color 2}
\gpsetlinetype{gp lt plot 2}
\draw[gp path] (10.663,1.362)--(11.579,1.362);
\draw[gp path] (1.504,2.505)--(1.609,2.559)--(1.715,2.612)--(1.820,2.666)--(1.926,2.720)%
  --(2.031,2.774)--(2.137,2.827)--(2.242,2.881)--(2.348,2.935)--(2.453,2.989)--(2.559,3.042)%
  --(2.664,3.096)--(2.770,3.150)--(2.875,3.204)--(2.981,3.257)--(3.086,3.311)--(3.192,3.365)%
  --(3.297,3.419)--(3.403,3.472)--(3.508,3.526)--(3.614,3.580)--(3.719,3.633)--(3.825,3.687)%
  --(3.930,3.741)--(4.036,3.795)--(4.141,3.848)--(4.247,3.902)--(4.352,3.956)--(4.458,4.010)%
  --(4.563,4.063)--(4.669,4.117)--(4.774,4.171)--(4.880,4.225)--(4.985,4.278)--(5.090,4.332)%
  --(5.196,4.386)--(5.301,4.440)--(5.407,4.493)--(5.512,4.547)--(5.618,4.601)--(5.723,4.647)%
  --(5.829,4.681)--(5.934,4.716)--(6.040,4.750)--(6.145,4.785)--(6.251,4.805)--(6.356,4.824)%
  --(6.462,4.844)--(6.567,4.864)--(6.673,4.884)--(6.778,4.903)--(6.884,4.923)--(6.989,4.943)%
  --(7.095,4.963)--(7.200,4.982)--(7.306,5.002)--(7.411,5.022)--(7.517,5.042)--(7.622,5.061)%
  --(7.728,5.081)--(7.833,5.107)--(7.939,5.138)--(8.044,5.169)--(8.150,5.200)--(8.255,5.230)%
  --(8.361,5.261)--(8.466,5.292)--(8.571,5.322)--(8.677,5.353)--(8.782,5.384)--(8.888,5.414)%
  --(8.993,5.445)--(9.099,5.476)--(9.204,5.507)--(9.310,5.537)--(9.415,5.568)--(9.521,5.599)%
  --(9.626,5.629)--(9.732,5.660)--(9.837,5.691)--(9.943,5.722)--(10.048,5.752)--(10.154,5.783)%
  --(10.259,5.814)--(10.365,5.844)--(10.470,5.875)--(10.576,5.906)--(10.681,5.937)--(10.787,5.967)%
  --(10.892,5.998)--(10.998,6.029)--(11.103,6.059)--(11.209,6.090)--(11.314,6.121)--(11.420,6.151)%
  --(11.525,6.182)--(11.631,6.213)--(11.736,6.244)--(11.842,6.274)--(11.947,6.305);
\gpcolor{color=gp lt color border}
\gpsetlinetype{gp lt border}
\draw[gp path] (1.504,7.825)--(1.504,0.985)--(11.947,0.985)--(11.947,7.825)--cycle;
\gpdefrectangularnode{gp plot 1}{\pgfpoint{1.504cm}{0.985cm}}{\pgfpoint{11.947cm}{7.825cm}}
\end{tikzpicture}
}
\caption{Running times of decremental single-source reachability algorithms dependent on the density of the initial graph. A point $ (\alpha, \beta) $ in this diagram means that for a graph with $ m = \Theta (n^\alpha) $ the algorithm has a running time of $ O (n^{\beta + o(1)}) $.}\label{fig:comparison_running_time}
\end{figure}

\paragraph*{Techniques.}
There are two novel technical contributions: (1) The algorithm of~\cite{HenzingerKNSTOC14} uses two kinds of randomly selected nodes, called {\em hubs} and {\em centers}, each fulfilling a different purpose. Maintaining an ES-tree for
each hub up to depth $h$, it
quickly tests for every pair of centers $ (x, y) $ whether there is a path of length at most $2h$ from $ x $ to $ y $ going through a hub. If there is no such path, we build a special graph, called {\em path union graph}, for the pair $(x, y)$ that contains all paths
of length $O(h)$ from $x$ to $y$. Since there no longer is a path from $x$ to $y$ through a hub of length at most $h$, we know that their path union graph is ``smaller'' than the original graph. In this paper we show how to
extend this approach multiple layers of path unions graphs.
Hubs and centers of the previous 
algorithms become level $k$,
resp.~$k - 1$ centers in the new approach. Level $k - 1$ centers serve as hubs for the level $k - 2$ centers, and more generally level $i$ centers serve as hubs for level $i-1$ centers. To do this efficiently we build the
ES-tree for a level $i$ center $x$ {\em  inside} the path-union graph of $x$ and another, potentially higher-level center. The fact that we use the smaller path-union graph instead of the original graph for these ES-trees
(together with an improved data structure for computing path-union graphs, see (2) below)
gives the improvement in the running time. 

(2) In~\cite{HenzingerKNSTOC14} we maintain for each center $x$ an approximate path union data structure that computes a superset of the path union of $x$ and any other center~$y$. This superset is an approximation
of the path union graph for $(x,y)$ as it
 might contain paths between the two centers of length $O(h \log n)$ (and {\em not} as desired $O(h)$), but no longer. The total time spent
 in this data structure for $x$ is (a) the size of the constructed path union graph and (b)
 a one-time ``global charge'' for using this data structure of
 $ O(n^2)$. It is based on a hierarchical graph decomposition technique.
Here we present a much simpler data structure that also constructs an approximate path union graph, but that does not require any hierarchical graph decomposition. This reduces the global charge per center from $O(n^2)$ to $O(m)$. 
We believe that this data structure is of independent interest.

\paragraph*{Outline.}
In \Cref{sec:prelim} we give the preliminaries.
In \Cref{sec:approx_path_union_ds} we present our new path union data structure.
Finally, in \Cref{sec:dense_SSR} we show how to combine this idea with the multi-layer path union approach to obtain a faster decremental single-source reachability algorithm for dense graphs.

\section{Preliminaries}\label{sec:prelim}

In this section we review some notions and basic facts that we will use in the rest of this paper.
We use the following notation:
We consider a directed graph $ G = (V, E) $ undergoing edge deletions, where $ V $ is the set of nodes of $ G $ and $ E $ is the set of edges of $ G $.
We denote by $ n $ the number of nodes of $ G $ and by $ m $ the number of edges of~$ G $ \emph{before the first edge deletion}.
For every pair of nodes $ u $ and $ v $ we denote the distance from $ u $ to $ v $ in~$ G $ by $ \dist_G (u, v) $.
For every subset of nodes $ U \subseteq V $, we define $ E (U) = E \cap U^2 $ and denote by $ G[U] = (U, E[U]) $ the \emph{subgraph of $ G $ induced by $ U $}.
For sets of nodes $ U \subseteq V $ and $ U' \subseteq V $ we define $ E (U, U') = E \cap (U \times U') $, i.e., $ E (U, U') $ is the set of edges $ (u, v) \in E $ such that $ u \in U $ and $ v \in U' $.
We write $ \hat{O} (T (m, n)) $ as an abbreviation for $ O (T (m, n) \cdot n^{o(1)}) $.

Like many decremental shortest paths and reachability algorithms, our algorithms internally use a data structure for maintaining a shortest paths tree up to a relatively small depth.

\begin{theorem}[Even-Shiloach tree~\cite{EvenS81,HenzingerK95,King99}]\label{lem:ES-tree}
There is a decremental algorithm, called \emph{Even-Shiloach tree} (short: ES-tree), that, given a directed graph~$ G $ undergoing edge deletions, a source node~$ s $, and a parameter $ h \geq 1 $, maintains a shortest paths tree from $ s $ and the corresponding distances up to depth $ h $ with total update time $ O (m h) $, i.e., the algorithm maintains $ \dist_G (s, v) $ and the parent of $ v $ in the shortest paths tree for every node $ v $ such that $ \dist_G (s, v) \leq h $.
By reversing the edges of $ G $ it can also maintain the distance from $ v $ to $ s $ for every node $ v $ in the same time.
\end{theorem}

The central concept in the algorithmic framework introduced in~\cite{HenzingerKNSTOC14} is the notion of the path union of a pair of nodes.
\begin{definition}\label{def:path_union_unweighted}
For every directed graph $ G $, every $ h\geq 1 $, and all pairs of nodes $ x $ and $ y $ of $ G $, the path union $ \cP (x, y, h, G) \subseteq V $ is the set containing all nodes that lie on some path $ \pi $ from $ x $ to $ y $ in $ G $ of weight at most $ h $.
\end{definition}
The path union has a simple characterization and can be computed efficiently.
\begin{lemma}[\cite{HenzingerKNSTOC14}]\label{pro:path_union_characterization}\label{lem:path_union_computation}
For every directed graph $ G $, every $ h \geq 1 $ and all pairs of nodes $ x $ and $ y $ of $ G $, we have $ \cP (x, y, h, G) = \{ v \in V \mid \dist_G (x, v) + \dist_G (v, y) \leq h \} $.
We can compute this set in time $ O (m) $.
\end{lemma}

Our algorithms use randomization in the following way: by sampling a set of nodes with a sufficiently large probability we can guarantee that certain sets of nodes contain at least one of the sampled nodes with high probability.
To the best of our knowledge, the first use of this technique in graph algorithms goes back to Ullman and Yannakakis~\cite{UllmanY91}.

\begin{lemma}\label{lem:center_on_shortest_path}\label{lem:hitting_set_argument}
Let $ T $ be a set of size $ t $ and let $ S_1, S_2, \ldots, S_k $ be subsets of $ T $ of size at least $ q $.
Let~$ U $ be a subset of $ T $ that was obtained by choosing each element of $ T $ independently with probability $ p = (a \ln{(k t)}) / q $, for some parameter $ a $.
Then, for every $ 1 \leq i \leq k $, the set $ S_i $ contains a node of $ U $ with high probability (whp), i.e., probability at least $ 1 - 1/t^a $, and the size of $ U $ is $ O ((t \log{(k t)}) / q) $ in expectation.
\end{lemma}

\section{Approximate Path Union Data Structure}\label{sec:approx_path_union_ds}

In this section we present a data structure for a graph $ G $ undergoing edge deletions, a fixed node~$ x $, and a parameter~$ h $.
Given a node $ y $, it computes an ``approximation'' of the path union $ \cP (x, y, h, G) $.
Using a simple static algorithm the path union can be computed in time $ O(m) $ for each pair $ (x, y) $.
We give an (almost) output-sensitive data structure for this problem, i.e., using our data structure the time will be proportional to the size of the approximate path union which might be $ o(m) $.
Additionally, we have to pay a global cost of $ O (m) $ that is amortized over \emph{all} approximate path union computations for the node $ x $ and \emph{all} nodes $ y $.
This will be useful because in our reachability algorithm we can use probabilistic arguments to bound the size of the approximate path unions.

\begin{proposition}\label{prop:approximate path union}
There is a data structure that, given a graph $ G $ undergoing edge deletions, a fixed node~$ x $, and a parameter~$ h $, provides a procedure \ApproximatePathUnion such that, given sequence of nodes $ y_1, \ldots, y_k $, this procedure computes sets $ F_1, \ldots F_k $ guaranteeing $ \cP (x, y, h, G) \subseteq F_i \subseteq \cP (x, y, (\log{m} + 3) h, G) $ for all $ 1 \leq i \leq k $.
The total running time is $ O(\sum_{1 \leq i \leq k} |F_i| + m) $.
\end{proposition}

\subsection{Algorithm Description}

Internally, the data structure maintains a set $ R (x) $ of nodes, initialized with $ R (x) = V $, such that the following invariant is fulfilled at any time:
all nodes that can be reached from $ x $ by a path of length at most~$ h $ are contained in $ R (x) $ (but $ R (x) $ might contain other nodes as well).
Observe that thus $ R (x) $ contains the path union $ \cP (x, y, h, G) $ for every node $ y $.

To gain some intuition for our approach consider the following way of computing an approximation of the path union $ \cP (x, y, h, G) $ for some node $ y $.
First, compute $ B_1 = \{ v \in R(x) \mid \dist_{G[R (x)]} (v, y) \leq h \} $ using a backward breadth-first search (BFS) to $ y $ in  $ G[R(x)] $, the subgraph of~$ G $ induced by $ R (x) $.
Second, compute $ F = \{ v \in R(x) \mid \dist_{G[B_1]} (x, v) \leq h \} $ using a forward BFS from $ x $ in~$ G[B_1] $.
It can be shown that $ \cP (x, y, h, G) \subseteq F \subseteq \cP (x, y, 2 h, G) $.\footnote{Indeed, $ F $ might contain some node $ v $ with $ \dist_G (x, v) = h $ and $ \dist_G (v, y) = h $, but it will not contain any node $ w $ with either $ \dist_G (x, w) > h $ or $ \dist_G (w, y) > h $.}
Given $ B_1 $, we could charge the time for computing $ F $ to the set $ F $ itself, but we do not know how to pay for computing $ B_1 $ as $ B_1 \setminus F $ might be much larger than $ F $.

Our idea is to additionally identify a set of nodes $ X \subseteq \{ v \in V \mid \dist_G (x, v) > h \} $ and remove it from $ R (x) $.
Consider a second approach where we first compute $ B_1 $ as above and then compute $ B_2 = \{ v \in R(x) \mid \dist_{G[R(x)]} (v, y) \leq 2 h \} $ and $ F = \{ v \in R(x) \mid \dist_{G[B_2]} (x, v) \leq h \} $.
It can be shown that $ \cP (x, y, h, G) \subseteq F \subseteq \cP (x, y, 3 h, G) $.
Additionally, all nodes in $ X = B_1 \setminus F $ are at distance more than~$ h $ from $ x $ and therefore we can remove $ X $ from $ R (x) $.
Thus, we can charge the work for computing $ B_1 $ and $ F $ to $ X $ and $ F $, respectively.\footnote{Note that in our first approach removing $ B_1 \setminus F $ would not have been correct as $ F $ was computed w.r.t\ to $ G[B_1] $ and not w.r.t.\ $ G[B_2] $.}
However, we now have a similar problem as before as we do not know whom to charge for computing $ B_2 $.

We resolve this issue by simply computing $ B_i = \{ v \in R(x) \mid \dist_{G[R (x)]} (v, y) \leq i h \} $ for increasing values of $ i $ until we arrive at some $ i^* $ such that the size of $ B_{i^*} $ is at most double the size of $ B_{i^*-1} $.
We then return $ F = \{ v \in R(x) \mid \dist_{G[B_i]} (x, v) \leq h \} $ and charge the time for computing $ B_i $ to $ X = B_{i-1} \setminus F $ and~$ F $, respectively.
As the size of $ B_i $ can double at most $ O (\log{n}) $ times we have $ \cP (x, y, h, G) \subseteq F \subseteq \cP (x, y, O (h \log{n}), G) $, as we show below.
Procedure~\ref{alg:path_unions} shows the pseudocode of this algorithm.
Note that in the special case that $ x $ cannot reach $ y $ the algorithm returns the empty set.
In the analysis below, let $ i^* $ denotes the final value of $ i $ before Procedure~\ref{alg:path_unions} terminates.

\begin{procedure}

\caption{ApproximatePathUnion($y$)}
\label{alg:path_unions}

\tcp{All calls of \ApproximatePathUnion{$y$} use fixed $ x $ and $ h $.}

Compute $ B_1 = \{ v \in R(x) \mid \dist_{G[R(x)]} (v, y) \leq h \} $ \tcp{backward BFS \emph{\textbf{to}} $ y $ in subgraph induced by $ R(x) $}
\For{$ i = 2 $ \KwTo $ \lceil \log{m} \rceil + 1 $}{
	Compute $ B_i = \{ v \in R(x) \mid \dist_{G[R(x)]} (v, y) \leq i h \} $ \tcp{backward BFS \emph{\textbf{to}} $ y $ in subgraph induced by $ R(x) $}
	\If{$ | E (B_i) | \leq 2 | E (B_{i-1}) | $}{\label{lin:check_for_size}
		Compute $ F = \{ v \in B_i \mid \dist_{G[B_i]} (x, v) \leq h \} $ \tcp{forward BFS \emph{\textbf{from}} $ x $ in subgraph induced by $ B_i $}
		$ X \gets B_{i-1} \setminus F $, $ R (x) \gets R (x) \setminus X $\;
		\Return{$ F $}\;
	}
}
\end{procedure}

\subsection{Correctness}
We first prove Invariant~(I): the set $ R (x) $ always contains all nodes that are at distance at most $ h $ from $ x $ in $ G $.
This is true initially as we initialize $ R (x) $ to be $ V $ and we now show that it continues to hold because we only remove nodes at distance more than $ h $ from $ x $.
\begin{lemma}
If $ R (x) \subseteq \{ v \in V \mid \dist_G (x, v) \leq h \} $, then for every node $ v \in X $ removed from $ R (x) $, we have $ \dist_G (x, v) > h $.
\end{lemma}

\begin{proof}
Let $ v \in X = B_{i^*-1} \setminus F $ and assume by contradiction that $ \dist_G (x, v) \leq h $.
Since $ v \in B_{i^*-1} $ we have $ \dist_{G[R(x)]} (v, y) \leq (i^*-1) h $.
Now consider the shortest path $ \pi $ from $ x $ to~$ v $ in~$ G $, which has length at most $ h $.
By the assumption, every node on $ \pi $ is contained in $ G[R (x)] $.
Therefore, for every node $ v' $ on $ \pi $, we have $ \dist_{G[R(x)]} (v', v) \leq h $ and thus
\begin{equation*}
\dist_{G[R(x)]} (v', y) \leq \dist_{G[R(x)]} (v', v) + \dist_{G[R(x)]} (v, y) \leq h + (i^*-1) h \leq i^* h
\end{equation*}
which implies that $ v' \in B_{i^*} $.
Thus, every node on $ \pi $ is contained in $ B_{i^*} $.
As $ \pi $ is a path from $ x $ to $ v $ of length at most $ h $ it follows that $ \dist_{G[B_{i^*}]} (x, v) \leq h $.
Therefore $ v \in F $, which contradicts the assumption $ v \in X $.
\end{proof}

We now complete the correctness proof by showing that the set of nodes returned by the algorithm approximates the path union.

\begin{lemma}
Procedure~\ref{alg:path_unions} returns a set of nodes $ F $ such that $ \cP (x, y, h, G) \subseteq F \subseteq \cP (x, y, (\log{m} + 3) h, G) $.
\end{lemma}

\begin{proof}
We first argue that the algorithm actually returns some set of nodes $ F $.
Note that in \Cref{lin:check_for_size} of the algorithm we always have $ | E (B_i) | \geq | E (B_{i-1}) | $ as $ B_{i-1} \subseteq B_i $.
As $ E (B_i) $ is a set of edges and the total number of edges is at most $ m $, the condition $ | E(B_i) | \leq | E (B_{i-1}) | $ therefore must eventually be fulfilled for some $ 2 \leq i \leq \lceil \log{m} \rceil + 1 $.

We now show that $ \cP (x, y, h, G) \subseteq F $.
Let $ v \in \cP (x, y, h, G) $, which implies that $ v $ lies on a path $ \pi $ from $ x $ to $ y $ of length at most $ h $.
For every node $ v' $ on $ \pi $ we have $ \dist_G (x, v') \leq h $, which by Invariant~(I) implies $ v' \in R (x) $.
Thus, the whole path $ \pi $ is contained in $ G[R (x)] $.
Therefore $ \dist_{G[R (x)]} (v', y) \leq h $ for every node $ v' $ on $ \pi $ which implies that $ \pi $ is contained in $ G[B_{i^*}] $.
Then clearly we also have $ \dist_{G[B_{i^*}]} (x, v) \leq h $ which implies $ v \in F $.

Finally we show that $ F \subseteq \cP (x, y, (\log{m} + 3) h, G) $ by proving that $ \dist_G (x, v) + \dist_G (v, y) \leq (\log{m} + 3) h $ for every node $ v \in F $.
As $ G [B_{i^*}] $ is a subgraph of~$ G $, we have $ \dist_G (x, v) \leq \dist_{G[B_{i^*}]} (x, v) $ and $ \dist_G (v, y) \leq \dist_{G[B_{i^*}]} (v, y) $.
By the definition of $ F $ we have $ \dist_{G[B_{i^*}]} (x, v) \leq h $.
As $ F \subseteq B_{i^*} $ we also have $ \dist_{G[B_{i^*}]} (v, y) \leq i^* h \leq (\lceil \log{m} \rceil + 1) h \leq (\log{m} + 2) h $.
It follows that $ \dist_G (x, v) + \dist_G (v, y) \leq {h + (\log{m} + 2) h} = {(\log{m} + 3) h} $.
\end{proof}

\subsection{Running Time Analysis}
To bound the total running time we prove that each call of Procedure~\ref{alg:path_unions} takes time proportional to the number of edges in the returned approximation of the path union plus the number of edges incident to the nodes removed from $ R (x) $.
As each node is removed from $ R(x) $ at most once, the time spent on \emph{all} calls of Procedure~\ref{alg:path_unions} is then $ O (m) $ plus the sizes of the subgraphs induced by the approximate path unions returned in each call.

\begin{lemma}\label{lem:running_time_approximate_path_union}
The running time of Procedure~\ref{alg:path_unions} is $ O (| E (F) | + | E (X, R(x)) | + | E (R(x), X) |) $ where $ F $ is the set of nodes returned by the algorithm, and $ X $ is the set of nodes the algorithm removes from $ R(x) $.
\end{lemma}

\begin{proof}
The running time in iteration $ 2 \leq j \leq i^*-1 $ is $ O(| E (B_j) |) $ as this is the cost of the breadth-first-search performed to compute $ B_j $.
In the last iteration $ i^* $, the algorithm additionally has to compute $ F $ and $ X $ and remove $ X $ from $ R (x) $.
As $ F $ is computed by a BFS in $ G[B_{i^*}] $ and $ X \subseteq B_{i^*-1} \subseteq B_{i^*} $, these steps take time $ O(| E (B_{i^*}) |) $.
Thus the total running time is $ O (\sum_{1 \leq j \leq i^*} | E (B_j) |) $.

By checking the size bound in Line~\ref{lin:check_for_size} of Procedure~\ref{alg:path_unions} we have $ |E (B_j) | > 2 | E (B_{j-1}) | $ for all $ 1 \leq j \leq i^*-1 $ and $ | E (B_{i^*}) | \leq 2 | E (B_{i^*-1}) | $.
By repeatedly applying the first inequality it follows that $ \sum_{1 \leq j \leq i^*-1} | E (B_j) | \leq 2 | E (B_{i^*-1}) | $.
Therefore we get
\begin{multline*}
\sum_{1 \leq j \leq i^*} | E (B_j) | = \sum_{1 \leq j \leq i^*-1} | E (B_j) | + | E (B_{i^*}) | \\
\leq 2 | E(B_{i^*-1}) | + 2 | E (B_{i^*-1}) | = 4 | E (B_{i^*-1}) |
\end{multline*}
and thus the running time is $ O (| E (B_{i^*-1}) |) $.
Now observe that by $ X = B_{i^*-1} \setminus F $ we have $ B_{i^*-1} \subseteq X \cup F $ and thus
\begin{align*}
E (B_{i^*-1}) &\subseteq E (F) \cup E (X) \cup E (X, F) \cup E (F, X) \\
 &\subseteq E (F) \cup E (X, R(x)) \cup E (R(x), X) \, .
\end{align*}
Therefore the running time is $ O (| E (F)| + | E (X, R(x)) | + | E (R(x), X) |) $.
\end{proof}

\section{Reachability via Center Graph}\label{sec:dense_SSR}

We now show how to combine the approximate path union data structure with a hierarchical approach to get an improved decremental reachability algorithm for dense graphs.
The algorithm has a parameter $ 1 \leq k \leq \log{n} $ and for each $ 1 \leq i \leq k $ a parameter $ c_i \leq n $.
We determine suitable choices of these parameters in \Cref{sec:running_time_dense}.
For each $ 1 \leq i \leq k-1 $, our choice will satisfy $ c_i \geq c_{i+1} $ and $ c_i = \hat{O} (c_{i+1}) $.
Furthermore, we set $ h_i = (3 + \log{m})^{i-1} n / c_1 $ for $ 1 \leq i \leq k $.
At the initialization, the algorithm determines sets of nodes $ C_1 \supseteq C_2 \supseteq \dots \supseteq C_k $ such that $ s, t \in C_1 $ as follows.
For each $ 1 \leq i \leq k $, we sample each node of the graph with probability $ a c_i \ln{n} / n $ (for a large enough constant $ a $), where the value of~$ c_i $ will be determined later.
The set $ C_i $ then consists of the sampled nodes, and if $ i \leq k-1 $, it additionally contains the nodes in $ C_{i+1} $.
For every $ 1 \leq i \leq k $ we call the nodes in $ C_i $ $i$-centers.
In the following we describe an algorithm for maintaining pairwise reachability between all $1$-centers.

\subsection{Algorithm Description}

\paragraph{Data Structures.}
The algorithm uses the following data structures:
\begin{itemize}
\item For every $i$-center $ x $ and every $ i \leq j \leq k $ an approximate path union data structure (see \Cref{prop:approximate path union}) with parameter $ h_j $.
\item For every $k$-center $ x $ an incoming and an outgoing ES-tree of depth $ h_k $ in $ G $.
\item For every pair of an $i$-center $ x $ and a $j$-center $ y $ such that $ l := \max(i, j) \leq k-1 $, a set of nodes $ \cQ (x, y, l) \subseteq V $.
Initially, $ \cQ (x, y, l) $ is empty and at some point the algorithm might compute $ \cQ (x, y, l) $ using the approximate path union data structure of $ x $.
\item For every pair of an $i$-center $ x $ and a $j$-center $ y $ such that $ l := \max(i, j) \leq k-1 $ an ES-tree of depth $ h_l $ from $ x $ in $ \cQ (x, y, l) $.
\item For every pair of an $i$-center $ x $ and a $j$-center $ y $ such that $ l := \max(i, j) \leq k-1 $ a set of $(l+1)$-centers certifying that $ x $ can reach $ y $.

\end{itemize}

\paragraph{Certified Reachability Between Centers (Links).}
The algorithm maintains the following limited path information between centers, called \emph{links}, in a top-down fashion.
Let $ x $ be a $k$-center and let $ y $ be an $i$-center for some $ 1 \leq i \leq k - 1 $.
The algorithm links $ x $ to $ y $ if and only if $ y $ is contained in the outgoing ES-tree of depth $ h_k $ of $ x $.
Similarly the algorithm links $ y $ to $ x $ if and only if $ y $ is contained in the incoming ES-tree of depth $ h_k $ of $ x $.
Let $ x $ be an $i$-center and let $ y $ be a $j$-center such that $ l := \max(i, j) \leq k-1 $.
If there is an $(l+1)$-center $ z $ such that $ x $ is linked to $ z $ and $ z $ is linked to $ y $, the algorithm links $ x $ to $ y $ (we also say that $ z $ links $ x $ to $ y $).
Otherwise, the algorithm computes $ \cQ (x, y, l) $ using the approximate path union data structure of $ x $ and starts to maintain an ES-tree from $ x $ up to depth $ h_l $ in $ G[\cQ (x, y, l)] $.
It links $ x $ to $ y $ if and only if $ y $ is contained in the ES-tree of $ x $.
Using a list of centers $ z $ certifying that $ x $ can reach $ y $, maintaining the links between centers is straightforward.

\paragraph{Center Graph.}
The algorithm maintains a graph called \emph{center graph}.
Its nodes are the $1$-centers and it contains the edge $ (x, y) $ if and only if $ x $ is linked to $ y $.
The algorithm maintains the transitive closure of the center graph.
A query asking whether a center $ y $ is reachable from a center $ x $ in $ G $ is answered by checking the reachability in the center graph.
As $ s $ and $ t $ are $1$-centers this answers $s$-$t$ reachability queries.

\paragraph*{Correctness.}
For the algorithm to be correct we have to show that there is a path from $ s $ to $ t $ in the center graph if and only if there is a path from $ s $ to $ t $ in $ G $.
We can in fact show more generally that this is the case for any pair of $1$-centers.

\begin{lemma}
For every pair of $1$-centers $ x $ and $ y $, there is a path from $ x $ to $ y $ in the center graph if and only if there is a path from $ x $ to $ y $ in $ G $.
\end{lemma}

\subsection{Running Time Analysis}\label{sec:running_time_dense}

The key to the efficiency of the algorithm is to bound the size of the graphs $ \cQ (x, y, l) $.

\begin{lemma}\label{lem:bound_on_size_of_path_union}
Let $ x $ be an $i$-center and let $ y $ be a $j$-center such that $ l := \max (i, j) \leq k-1 $.
If $ x $ is not linked to $ y $ by an $(l+1)$-center, then $ \cQ (x, y, l) $ contains at most $ n / c_{l+1} $ nodes with high probability.
\end{lemma}
With the help of this lemma we first analyze the running time of each part of the algorithm and argue that our choice of parameters gives the desired total update time.

\paragraph{Parameter Choice.}
We carry out the running time analysis with regard to two parameters $ 1 \leq b \leq c \leq n $ which we will set at the end of the analysis.
We set
$ k = \lceil (\log{(c/b)}) / (\sqrt{\log{n} \cdot \log{\log{n}}} ) \rceil + 1 $,
$ c_k = b $ and $ c_{i} = 2^{\sqrt{\log{n} \cdot \log{\log{n}}}} c_{i+1} = \hat{O} (c_{i+1}) $ for $ 1 \leq i \leq k-1 $.
Note that the number of $i$-centers is $ \tilde O (c_i) $ in expectation.
Observe that
\begin{multline*}
(3 + \log{m})^{k-1} = O ((\log{n})^k) \leq O ((\log{n})^{\sqrt{\log{n} / \log{\log{n}}}}) \\ = O (2^{\sqrt{\log{n} \cdot \log{\log{n}}}}) = O (n^{\sqrt{\log{\log{n}} / \log{n}}}) = O (n^{o(1)}) \, .
\end{multline*}
Furthermore we have
\begin{equation*}
c_1 = \left( 2^{\sqrt{\log{n} \cdot \log{\log{n}}}} \right)^{k-1} c_k \geq 2^{\log{(c / b)}} b = \frac{c}{b} \cdot b = c
\end{equation*}
and by setting $ k' = (\log{(c/b)}) / (\sqrt{\log{n} \cdot \log{\log{n}}}) $ we have $ k \leq k' + 2 $ and thus
\begin{equation*}
c_1 = \left( 2^{\sqrt{\log{n} \cdot \log{\log{n}}}} \right)^{k-1} c_k \leq \left( 2^{\sqrt{\log{n} \cdot \log{\log{n}}}} \right)^{k'+1} c_k = 2^{\sqrt{\log{n} \cdot \log{\log{n}}}} c = \hat O (c) \, .
\end{equation*}
Remember that $ h_i = (3 + \log{m})^{i-1} n / c_1 $ for $ 1 \leq i \leq k $.
Therefore we have $ h_i = \hat{O} (n / c_1) = \hat{O} (n / c) $.

\paragraph{Maintaining ES-Trees.}
For every $k$-center we maintain an incoming and an outgoing ES-tree of depth $ h_k $, which takes time $ O (m h_k) $.
As there are $ \tilde O (c_k) $ $k$-centers, maintaining all these trees takes time $ \tilde O (c_k m h_k) = \hat O (b m n / c) $.

For every $i$-center $ x $ and every $j$-center $ y $ such that $ l := \max (i, j) \leq k-1 $, we maintain an ES-tree up to depth $ h_l $ in $ G[\cQ (x, y, l)] $.
By \Cref{lem:bound_on_size_of_path_union} $ \cQ (x, y, l) $ has at most $ n / c_{l+1} $ nodes and thus $ G[\cQ (x, y, l)] $ has at most $ n^2 / c_{l+1}^2 $ edges.
Maintaining this ES-tree therefore takes time $ O ((n^2 / c_{l+1}^2) \cdot h_l) = \hat O (n^2 / c_{l+1}^2 (n / c_1)) = \hat O (n^3 / (c_1 c_{l+1}^2)) $.
In total, maintaining all these trees takes time
\begin{multline*}
\hat{O} \left( \sum_{1 \leq i \leq k-1} \sum_{1 \leq j \leq i} c_i c_j \frac{n^3}{c_1 c_{i+1}^2} \right) =
\hat{O} \left( \sum_{1 \leq i \leq k-1} \sum_{1 \leq j \leq i} \frac{c_i c_1 n^3}{c_{i+1} c_1 c_k} \right) \\
= \hat{O} \left( \sum_{1 \leq i \leq k-1} \sum_{1 \leq j \leq i} \frac{n^3}{c_k} \right)
= \hat{O} \left( k^2 \frac{n^3}{c_k} \right) = \hat{O} \left(\frac{n^3}{b} \right) \, .
\end{multline*}

\paragraph{Computing Approximate Path Unions.}
For every $i$-center $ x $ and every $ i \leq j \leq k $ we maintain an approximate path union data structure with parameter $ h_j $.
By \Cref{prop:approximate path union} this data structures has a total running time of $ O (m) $ and an additional cost of $ O (|E (\cQ (x, y, j)) |) $ each time the approximate path union $ \cQ (x, y, j) $ is computed for some $j$-center $ y $.
By \Cref{lem:bound_on_size_of_path_union} the number of nodes of $ \cQ (x, y, j) $ is $ n / c_{j+1} $ with high probability and thus its number of edges is $ n^2 / c_{j+1}^2 $.
Therefore, computing all approximate path unions takes time
\begin{multline*}
\tilde O \left( \sum_{1 \leq i \leq k-1} \sum_{i \leq j \leq k} \left( c_i m + c_i c_j \frac{n^2}{c_{j+1}^2} \right) \right) \! =
\tilde O \left( \sum_{1 \leq i \leq k-1} \sum_{i \leq j \leq k} \left( c_1 m + \frac{c_1 c_j n^2}{c_{j+1} c_k} \right) \right) \\
= \hat O \left( \sum_{1 \leq i \leq k-1} \sum_{i \leq j \leq k} \left( c_1 m + \frac{c_1 n^2}{c_k} \right) \right)
\! = \hat O ( k^2 c_1 m + k^2 c_1 n^2 / c_k )
\! = \hat O ( c m + c n^2 / b ) \, .
\end{multline*}

\paragraph{Maintaining Links Between Centers.}
For each pair of an $i$-center $ x $ and a $j$-center $ y $ there are at most $ \tilde O (c_{l+1}) $ $(l+1)$-centers that can possibly link $ x $ to $ y $.
Each such $(l+1)$-center is added to and removed from the list of $(l+1)$-centers linking $ x $ to $ y $ at most once.
Thus, the total time needed for maintaining all these links is $ \tilde O (\sum_{1 \leq i \leq k-1} \sum_{1 \leq j \leq i} c_i c_j c_{i+1} ) = \tilde O (k^2 c_1^3) = \tilde O (c^3) $.

\paragraph{Maintaining Transitive Closure in Center Graph.}
The center graph has $ \tilde O (c_1) $ nodes and thus $ \tilde O (c_1^2) $ edges.
During the algorithm edges are only deleted from the center graph and never inserted.
Thus we can use known $ O(mn) $-time decremental algorithms for maintaining the transitive closure~\cite{RodittyZ08,Lacki13} in the center graph in time $ \tilde O (c_1^3) = \tilde O (c^3) $.

\paragraph{Total Running Time.}
Since the term $ c n^2 / b$ is dominated by the term $ n^3 / b $, we obtain a total running time of $\hat{O} \left( b m n / c + n^3 / b + c m + c^3 \right) $.
By setting $ b = n^{5/3} / m^{2/3} $ and $ c = n^{4/3} / m^{1/3} $ the running time is $ \hat{O} (m^{2/3} n^{4/3} + n^4 / m) $ and by setting $ b = n^{9/7} / m^{3/7} $ and $ c = m^{1/7} n^{4/7} $ the running time is $ \hat{O} (m^{3/7} n^{12/7} + m^{8/7} n^{4/7}) $.

\subsection{Decremental Single-Source Reachability}

The algorithm above works for a set of randomly chosen centers.
Note that the algorithm stays correct if we add any number of nodes to $ C_1 $, thus increasing the number of $1$-centers for which the algorithm maintains pairwise reachability.
If the number of additional centers does not exceed the expected number of randomly chosen centers, then the same running time bounds still apply.
Using the reductions of~\cite{HenzingerKNSTOC14} this immediately implies decremental algorithms for maintaining single-source reachability and strongly connected components.

\begin{theorem}
There are decremental algorithms for maintaining single-source reachability and strongly connected components with constant query time and expected total update time
$ \hat{O} (m^{2/3} n^{4/3} + m^{3/7} n^{12/7}) $
that are correct with high probability against an oblivious adversary.
\end{theorem}

\printbibliography[heading=bibintoc] 
\end{document}